\documentclass[10pt]{article}
\usepackage[margin=1in]{geometry}
\usepackage{amsmath,amssymb,amsthm}
\usepackage{graphicx}
\usepackage{booktabs}
\usepackage{url}
\usepackage[hidelinks]{hyperref}

\usepackage{expl3}
\ExplSyntaxOn
\prop_new:N \g_num_prop
\NewDocumentCommand{\defnum}{mm}{\prop_gput:Nnn \g_num_prop {#1} {#2}}
\NewDocumentCommand{\num}{m}{\prop_item:Nn \g_num_prop {#1}}
\ExplSyntaxOff
\defnum{build/deep100m-graft-quality-s}{19,675}
\defnum{build/deep100m-misi-s}{5,250}
\defnum{build/ratio-graft-quality}{3.7}
\defnum{build/sift-misi-s}{12}
\defnum{deep100m/hnsw/qps@0.99}{4,513}
\defnum{e5/diskann-12g-qps@0.987}{1,983}
\defnum{e5/misi-8g-qps}{8}
\defnum{e5/misi-8g-qps-high}{30}
\defnum{e5/misi-8g-qps-low}{8}
\defnum{e5/misi-8g-recall}{0.9936}
\defnum{sat/diskann-max}{0.9984}
\defnum{sat/graft-max}{0.9992}
\defnum{sat/hnsw-max}{0.9989}
\defnum{sat/misi-max}{0.9967}
\defnum{scaling/c99-100M}{20,000}
\defnum{scaling/c99-10M}{10,000}
\defnum{scaling/c99-1M}{5,000}
\defnum{scaling/gamma}{0.30}
\defnum{sift/cov-gain@C500}{0.037}
\defnum{sift/hnsw-over-misi@0.99}{6}
\defnum{sift/hnswlib/qps@0.99}{15,751}
\defnum{sift/misi-over-napp@0.99}{2.2}
\defnum{sift/misi/qps@0.99}{2,637}
\defnum{sift/napp/qps@0.99}{1,207}

\newtheorem{theorem}{Theorem}
\newtheorem{lemma}{Lemma}
\newtheorem{corollary}{Corollary}
\newtheorem{remark}{Remark}
\newtheorem{assumption}{Assumption}
\newtheorem{observation}{Observation}

\newcommand{\misi}{\textsf{misi}}
\newcommand{\kb}{k_b}
\newcommand{\ks}{k_s}

\title{\misi: a Metric Inverted Sample Index}
\author{Edgar Ch\'avez\\ CICESE\\ \texttt{elchavez@cicese.mx}}
\date{}

\begin{document}
\maketitle

\begin{abstract}
We present \misi, an inverted index for approximate nearest-neighbor
search over general metric spaces whose vocabulary is a random sample of
the database itself, of size proportional to $n$. Each object is
represented by its $\kb$ nearest sample points, found by a pluggable
inner index over the sample; queries are answered by an idf-weighted
shared-neighbor vote followed by exact verification of $C$ candidates.
The construction generalizes the neighborhood-approximation (NAPP) index
of Tellez, Ch\'avez and Navarro from a constant number of pivots to a
linear-size vocabulary, which keeps posting lists at a constant expected
length $\rho = \kb/\alpha$ as $n$ grows and turns the index into a
combinator: any high-recall index on $\alpha n$ points yields an index
on $n$ points, for any metric. A probabilistic model of the vote gives a
recall guarantee --- $\kb$ logarithmic in $n$ over the overlap gap
suffices, with a verification budget equal to a per-query confusable
count that the index itself estimates --- and a matching limit: no
processing of the vote resolves overlap differences below order
$1/\sqrt{\kb}$. The design's strengths are structural: construction is
$n$ independent searches --- embarrassingly parallel, deterministic at
any thread count, \num{build/deep100m-misi-s}\,s for $10^8$ vectors on
64 cores, \num{build/ratio-graft-quality}$\times$ faster than a
matched-recall graph build --- it streams under an enforced $3$\,GiB
memory cap, and the artifact is a portable directory of flat arrays that
serves $10^8$ vectors from NVMe within an enforced $8$\,GB budget,
below the working floor of the SSD-graph baseline. Its cost is
query-time work: at matched recall in RAM, saturated graph baselines
answer $6$--$16\times$ faster at $10^6$ and reach marginally higher
recall at $10^8$, and the verification budget for $0.99$ recall grows
as $n^{\num{scaling/gamma}}$. All results carry seeds, saturation
sweeps and full configurations, are generated from run manifests, and
include a catalogue of measured negative results delimiting what
query-time processing of the stored evidence can add. The intended
applications are those where construction cost, determinism, memory
footprint, or a black-box metric dominate the requirements ---
frequently rebuilt corpora, batch similarity workloads, and
constrained-memory serving --- rather than peak query throughput.
\end{abstract}

\section{Introduction}

Metric indexes classically summarize each object by its distances to a
small, fixed set of reference points. The permutation-index family
\cite{chavez2008permutations,amato2008mifile,esuli2012ppindex} replaced
distances by the \emph{order} of references, and its most scalable member,
the neighborhood-approximation index (NAPP) of Tellez, Ch\'avez and
Navarro \cite{tellez2011sisap,tellez2013napp}, kept only each object's $K$
nearest references out of $m$, stored them inverted, and answered queries
by counting shared references and verifying candidates. Naidan, Boytsov
and Nyberg \cite{naidan2015permutation} later showed this family
competitive with the indexes of its day.

This paper asks what happens when NAPP's vocabulary stops being a design
constant and becomes a \emph{sample of the database itself}, of size
$m = \alpha n$ for a fixed rate $\alpha$. The resulting index is
\misi.\footnote{The repository name expands MISIFU as \emph{Metric
Inverted Sample-Index For Ultrascale}; \emph{misi} is also how Spanish
speakers call a cat, and Misif\'u the archetypal cat so called. Both
etymologies are correct.} Two consequences are immediate.
First, the reference search can no longer be brute force --- at
$m = 2{\cdot}10^6$ references, computing each object's nearest references
requires an approximate-nearest-neighbor structure of its own, so the
design becomes a \emph{combinator}: it turns any high-recall index on
$\alpha n$ points into an index on $n$ points, for any metric, since the
core consumes only neighbor identities. Second, the posting lists stop
growing with $n$: at fixed dials their mean length is pinned at
$\rho = \kb/\alpha$, whereas a fixed vocabulary's lists grow linearly.
The experiments in \S\ref{sec:experiments} quantify both effects; the
comparison against NAPP itself, run under one protocol on one machine, is
reported first.

The resulting structure has properties that graph indexes do not offer:
construction is $n$ independent searches against a frozen sample index
(embarrassingly parallel, deterministic at any thread count, and
streamable under a measured $3$\,GiB memory cap); the artifact is a
directory of flat arrays that can be built on one machine and served on
another; and the resident set for serving is governed by the sample, not
the database. It also has a cost that graph indexes do not pay: every
query streams $\sim\!\ks\rho$ postings through an accumulator, and at
matched recall on data that fits in RAM, saturated beam search is
$6$--$16\times$ faster (\S\ref{sec:experiments}). Both facts are
measured and reported with their configurations.

\paragraph{Where these trade-offs pay.} Several deployment patterns
weight the axes the way this design does. \emph{Frequently rebuilt
corpora}: when a collection changes enough that the index is rebuilt
nightly or hourly (catalogues, feeds, logs), construction cost dominates
the total cost of ownership, and an $87$-minute deterministic build of
$10^8$ vectors --- streamable on a small machine --- substitutes for
update machinery entirely. \emph{Batch similarity workloads}:
similarity joins, deduplication sweeps, and $k$NN-graph construction
for downstream learning issue all their queries at once; throughput at
batch scale is the relevant metric, latency is not, and the eager build
amortizes after one block of queries. \emph{Constrained-memory
serving}: personal or on-device corpora served from flash within a few
GB of RAM, at tens of queries per second, in the budget window measured
in \S\ref{sec:experiments} where the SSD-graph baseline does not run.
\emph{Black-box metrics}: the core consumes neighbor identities only,
so any dissimilarity with a high-recall inner index indexes without
modification --- a design property, evaluated here on vectors only.
\emph{Reproducible pipelines}: bit-identical artifacts across thread
counts and build modes simplify auditing and regression testing in
settings where an index is a versioned artifact. None of these
scenarios requires winning the in-RAM QPS comparison, and this paper
does not claim to.

\paragraph{Contributions.}
\begin{enumerate}
\item The structure and its two-dial account --- resolution
  $\rho = \kb/\alpha$, confidence $\kb$ --- with a model
  (\S\ref{sec:model}) giving a recall guarantee
  (Theorem~\ref{thm:vote}) whose budget is an instance quantity the
  index itself estimates, and a resolution limit
  (Lemma~\ref{lem:limit}) for the vote.
\item Measurements under a fixed protocol (seeds, saturation sweeps,
  full configurations per row, numbers generated from run manifests):
  the NAPP and \texttt{hnswlib} comparison at $10^6$
  (Table~\ref{tab:e1}); the scaling law
  $C_{0.99} \propto n^{\num{scaling/gamma}}$ over nested Deep prefixes
  (Table~\ref{tab:scale}); saturated ceilings for four systems at
  $10^8$ (Table~\ref{tab:sat}); and memory-limited serving under
  enforced cgroup budgets (Table~\ref{tab:memory}).
\item A set of scoped negative results (\S\ref{sec:negative}): three
  query-time refinement families, per-query adaptive stopping, distance
  bounds computable from stored ranks (impossible without a stored
  scale; Remark~\ref{rem:noscale}), and posting subsampling
  (dominated by the native dials at equal evidence read).
\item Coverage-driven vocabulary re-sampling with seeds
  (\S\ref{sec:coverage}): a large, stable gain at small verification
  budgets on SIFT ($+\num{sift/cov-gain@C500}$ recall at $C{=}500$),
  shrinking with $C$, and reversing sign on GloVe at large $C$ --- the
  conditions under which it helps and hurts are stated.
\item The measured systems properties: deterministic parallel
  construction ($\num{build/ratio-graft-quality}\times$ the strongest
  graph build at $10^8$), a streaming build under an enforced
  $3$\,GiB cap, build-here/serve-there artifacts, and disk-resident
  serving within $8$\,GB where the SSD-graph baseline does not run.
\end{enumerate}

Everything unmeasured is labeled future work, including the two
compositions the analysis motivates most strongly: a resident
product-quantization screen before exact verification, and recursive
construction (the inner index built by the same method).

\section{Related work}\label{sec:related}

\paragraph{Permutation and reference-based indexing.} Ordering
permutations \cite{chavez2008permutations}; the MI-File
\cite{amato2008mifile,amato2014mifile}, whose later experiments used
reference sets in the tens of thousands; prefix permutations
\cite{esuli2012ppindex}; and NAPP
\cite{tellez2011sisap,tellez2013napp}, the direct ancestor: $K$ nearest
pivots of $m$, inverted, candidates by shared-pivot count with a
threshold, then verification. \misi{} differs in the vocabulary size
($\Theta(n)$ versus a constant), the pivot search (a pluggable ANN index
versus exhaustive), and the candidate rule (a weighted vote and a budget
$C$ versus a count threshold $t$). The pivot-selection literature
\cite{bustos2003pivots} is the classical treatment of the vocabulary
question that \S\ref{sec:coverage} revisits at linear scale.

\paragraph{Graph indexes.} HNSW \cite{malkov2020hnsw}, DiskANN
\cite{jayaram2019diskann}, and GRAFT \cite{graft}, which we use both as
a baseline and as the inner index; its construction is deterministic and
embarrassingly parallel, which the combinator inherits. FreshDiskANN
\cite{freshdiskann} and SPFresh \cite{spfresh} handle updates for
disk-resident graphs; incremental insertion is native to navigable
graphs.

\paragraph{Partitioned inverted designs.} SPANN \cite{spann} is the
closest system by architecture: a vocabulary of actual data points at a
scale of $n/10$ or larger, an ANN index over the vocabulary in RAM,
posting lists on disk, multi-assignment, and disk-resident serving at
billion scale. \misi{} differs by sampling instead of clustering, by
$\kb = 64$ assignments instead of $\le 8$, and by an idf-weighted vote
instead of union-and-rerank. A controlled comparison is the most
important experiment this paper does not contain; it is in progress for
the full version, and no claim here should be read as standing in for
it. FAISS's IVF with an HNSW quantizer realizes the same
coarse-quantizer-with-ANN-vocabulary idea at smaller vocabularies.

\paragraph{Shared neighbors.} Shared-nearest-neighbor similarity
originates with Jarvis and Patrick \cite{jarvis1973} and was developed
for clustering \cite{ertoz2003}; \misi{} uses it as the retrieval signal
itself. The merge is the $t$-occurrence problem of approximate string
joins \cite{li2008divideskip}, whose skipping algorithms are discussed
in \S\ref{sec:merge}.

\section{The index}\label{sec:index}

\paragraph{Structure.} Draw $S \subseteq D$, $|S| = m = \alpha n$,
uniformly at random (\S\ref{sec:coverage} treats non-uniform selection).
Build a high-recall inner index $I_S$ over $S$; freeze it. For every
object $o$, search $I_S$ for the $\kb$ nearest sample points of $o$ ---
its \emph{signature} $N_S(o)$, stored with build ranks. Invert: the
posting list $L(s)$ holds every object with $s$ in its signature,
sorted by (rank, id). Mean list length is exactly $\rho = \kb/\alpha$
(there are $\kb n$ postings over $\alpha n$ lists), independent of $n$.
The signatures double as a forward file. Construction is $n$ independent
inner searches plus a linear scatter: embarrassingly parallel,
deterministic at any thread count, and identical across the in-memory,
streaming, and factory build modes (enforced by test; the streaming mode
is measured in \S\ref{sec:experiments}).

\paragraph{Query.} Given $q$ and target $k$: (1) search $I_S$ for the
$\ks$ nearest sample points of $q$ ($\ks$ is a query-time knob,
decoupled from $\kb$); (2) merge the $\ks$ posting lists, scoring each
candidate by its idf-weighted shared-neighbor count; (3) verify the
top-$C$ candidates with true distances and return the best $k$. The two
query dials are $\ks$ (how many lists) and $C$ (how many true
distances); \S\ref{sec:experiments} shows which one to move for a given
recall target.

\paragraph{Residency.} Serving requires in RAM: the sample vectors, the
inner index, and the offset table. Postings and database vectors may be
memory-mapped from disk; the merge then uses an
$O(\ks\rho)$-transient accumulator rather than a dense scoreboard, and
the vector mapping must be advised for random access (without
\texttt{MADV\_RANDOM}, fault-around read-amplifies the $384$-byte
verification reads by more than two orders of magnitude). The measured
budget floor of the current implementation is set by the inner-index
build at load time, not by the sample (Table~\ref{tab:memory}).

\paragraph{Updates (design, not measured).} Because the sample is a
frozen statistical summary, an insert is one inner search plus $\kb$
posting appends, and a delete is an exact posting removal via the
forward file. None of this is implemented or benchmarked here ---
the rank-sorted list layout conflicts with cheap appends and would need
a bucketed layout first --- so this paper claims only the design
property and defers measurement. Navigable graphs insert incrementally
as a matter of course; the comparison, when made, must be against
\cite{freshdiskann,spfresh} and \texttt{hnswlib}'s native inserts.

%

\section{A model of the vote}\label{sec:model}

The design has two dials --- resolution $\rho = \kb/\alpha$ and confidence
$\kb$ --- and one architectural commitment: candidates are ranked by a
shared-neighbor vote and resolved by verification. This section proves what
the vote can and cannot do. The positive result
(Theorem~\ref{thm:vote}) says the vote separates candidates whose
neighborhood overlap with the query differs by $\Delta$, at confidence
$\kb = O(\log(kn/\delta)/\Delta^2)$, so vote-then-verify succeeds with a
budget governed by a single instance quantity: the number of database
objects whose overlap with the query exceeds a threshold. The negative
result (Lemma~\ref{lem:limit}) says the vote resolves overlap differences
of order $1/\sqrt{\kb}$ and no finer --- the provable form of what our
earlier experiments observed as a ``resolution barrier,'' now scoped to
what is actually established. Both statements are about the
\emph{information in the stored evidence}; neither claims optimality among
all conceivable uses of it (\S\ref{sec:evidence-scope}).

\subsection{Setting}

$D$ is the database, $|D| = n$. For $x \in D \cup \{q\}$ and integer
$\rho$, let $N_\rho(x) \subseteq D$ be the $\rho$ nearest objects to $x$
(ties broken by identifier). $S \subseteq D$ is a Bernoulli sample: each
object enters $S$ independently with probability $\alpha$. Define the
\emph{$\rho$-trace} signature
\[
\tilde N_S(x) \;=\; S \cap N_\rho(x),
\qquad |\tilde N_S(x)| \sim \mathrm{Bin}(\rho, \alpha),\quad
\mathbb{E}|\tilde N_S(x)| = \alpha\rho = \kb .
\]
The implemented signature truncates at exactly $\kb$ sample points rather
than at radius rank $\rho$; Remark~\ref{rem:truncation} transfers the
results between the two at an additive $e^{-\Omega(\kb\varepsilon^2)}$.

The \emph{overlap} of two objects is the normalized intersection of their
$\rho$-neighborhoods,
\[
\omega(q,o) \;=\; \frac{|N_\rho(q) \cap N_\rho(o)|}{\rho} \;\in\; [0,1],
\]
and the \emph{raw vote} is the sampled intersection,
\[
X(q,o) \;=\; \bigl|S \cap N_\rho(q) \cap N_\rho(o)\bigr|,
\qquad \mathbb{E}\,X(q,o) = \kb\,\omega(q,o).
\]
Two remarks before the mathematics. First, $X/\kb$ is an unbiased
estimator of $\omega$: \emph{the index measures its own hardness}, a fact
used below to make the theorem's instance quantity observable. Second,
every statement in this section conditions on the neighborhoods
$N_\rho(\cdot)$ --- the geometry is arbitrary; only the sample is random.
This is what makes the analysis metric-free: nothing about the distance
enters except through the sets $N_\rho$.

\subsection{Separation}

\begin{lemma}[vote separation]\label{lem:sep}
Let $o, o'$ satisfy $\omega(q,o) \ge \omega(q,o') + \Delta$ with
$\Delta > 0$. Then
\[
\Pr\bigl[X(q,o') \ge X(q,o)\bigr]
\;\le\;
\exp\!\Bigl(-\,\frac{\kb\,\Delta^2}{2\,(1 + \Delta/3)}\Bigr).
\]
\end{lemma}

\begin{proof}
Write $A = N_\rho(q) \cap N_\rho(o) \setminus N_\rho(o')$ and
$B = N_\rho(q) \cap N_\rho(o') \setminus N_\rho(o)$; $A$ and $B$ are
disjoint subsets of $N_\rho(q)$, so $|A| + |B| \le \rho$. With
$B_u \sim \mathrm{Bernoulli}(\alpha)$ i.i.d.\ over $u$,
\[
Z \;:=\; X(q,o) - X(q,o')
\;=\; \sum_{u \in A} B_u \;-\; \sum_{u \in B} B_u ,
\]
a sum of $|A| + |B|$ independent summands with values in $\{-1, 0, 1\}$,
mean $\mathbb{E}Z = \alpha(|A| - |B|) = \kb\bigl(\omega(q,o) -
\omega(q,o')\bigr) \ge \kb\Delta$, and total variance
$\sum \mathrm{Var} \le \alpha(|A| + |B|) \le \alpha\rho = \kb$. Bernstein's
inequality for the lower tail at $t = \kb\Delta$ gives
$\Pr[Z \le 0] \le \Pr[Z - \mathbb{E}Z \le -t] \le
\exp(-t^2 / (2(\kb + t/3)))$, which is the claim.
\end{proof}

\begin{theorem}[vote-then-verify]\label{thm:vote}
Fix a query $q$, a target $k$, and a threshold $\omega^- \in [0,1)$.
Define the \emph{confusable set}
\[
K(q, \omega^-) \;=\; \{\, o \in D : \omega(q,o) > \omega^- \,\}
\]
and suppose every true $k$-nearest neighbor $o^*$ of $q$ satisfies
$\omega(q,o^*) \ge \omega^- + \Delta$. If
\[
\kb \;\ge\; \frac{2\,(1 + \Delta/3)}{\Delta^2}\,
\ln\!\frac{k\,n}{\delta},
\]
then with probability at least $1 - \delta$ every true $k$-nearest
neighbor is among the top $|K(q,\omega^-)| + k$ candidates in vote order.
Consequently vote-then-verify with budget $C = |K(q,\omega^-)| + k$
returns the exact $k$ nearest neighbors with probability $\ge 1 - \delta$.
\end{theorem}

\begin{proof}
Fix a true neighbor $o^*$ and any $o \notin K(q,\omega^-)$; their overlaps
differ by at least $\Delta$, so by Lemma~\ref{lem:sep} the event
$X(q,o) \ge X(q,o^*)$ (which covers ties, hence any tie-breaking rule) has
probability at most $\exp(-\kb\Delta^2/(2(1+\Delta/3))) \le \delta/(kn)$
by the choice of $\kb$. A union bound over the $k$ true neighbors and the
at most $n$ objects outside $K$ leaves, with probability $\ge 1-\delta$,
no object outside $K$ ranked at or above any true neighbor. On that event
each true neighbor is outranked only by members of $K$ and by other true
neighbors: its vote rank is at most $|K| + k$. Verification computes true
distances for the top $C = |K| + k$ and returns the exact $k$-NN.
\end{proof}

Three comments. (i) The budget is an \emph{instance} quantity: $C$
succeeds precisely when it covers the confusable set, and
$|K(q,\omega^-)|$ varies per query --- the empirical spread of adaptive
budgets (Fig.~\ref{fig:hardness}) is this quantity made visible. (ii) The
theorem prescribes $\kb = \Theta(\log(n)/\Delta^2)$ for fixed failure
probability --- logarithmic confidence, with the constant set by the
overlap gap $\Delta$, a property of the data at resolution $\rho$. (iii)
Nothing was assumed about the metric: expansion or doubling conditions
would be needed only to bound $|K(q,\omega^-)|$ a priori, which we do not
attempt; we \emph{measure} it instead (via $X/\kb$), and
\S\ref{sec:experiments} reports the measured budgets.

\subsection{The resolution limit}

The converse question: which overlap differences can the vote
\emph{not} see?

\begin{assumption}[generic overlap]\label{as:generic}
For the pair $(o, o')$ under comparison there is a constant
$\beta \in (0,1]$ with
$|N_\rho(q) \cap (N_\rho(o) \,\triangle\, N_\rho(o'))| \ge \beta\rho$:
the two candidates disagree on a constant fraction of the query's
$\rho$-neighborhood. In high intrinsic dimension this is the typical
case; it fails only when $o$ and $o'$ keep nearly identical company, in
which case the vote (correctly) cannot and need not distinguish them.
\end{assumption}

\begin{lemma}[resolution limit of the vote]\label{lem:limit}
Under Assumption~\ref{as:generic}, there are constants $c_0, c_1 > 0$
(depending only on $\beta$ and $\alpha \le 1/2$) such that if
\[
0 \;\le\; \omega(q,o) - \omega(q,o') \;\le\; \frac{c}{\sqrt{\kb}},
\]
then
\[
\Pr\bigl[X(q,o') \ge X(q,o)\bigr] \;\ge\; \frac12 - c_0\,c -
\frac{c_1}{\sqrt{\kb}} .
\]
The vote misranks the pair with probability bounded away from zero: an
overlap advantage of order $1/\sqrt{\kb}$ is invisible to it.
\end{lemma}

\begin{proof}
$Z = X(q,o) - X(q,o')$ is, as in Lemma~\ref{lem:sep}, a sum of
$N = |A| + |B| \ge \beta\rho$ independent bounded summands with mean
$\mu = \kb(\omega(q,o) - \omega(q,o')) \le c\sqrt{\kb}$ and variance
$\sigma^2 \ge \alpha(1-\alpha)\,\beta\rho = (1-\alpha)\beta\,\kb$. By the
Berry--Esseen theorem,
$\Pr[Z \le 0] \ge \Phi(-\mu/\sigma) - C_{\mathrm{BE}}\,\gamma$, where
$\gamma = O(1/\sigma) = O(1/\sqrt{\kb})$ bounds the normalized third
moments. Since $\mu/\sigma \le c/\sqrt{(1-\alpha)\beta}$ and
$\Phi(-x) \ge \tfrac12 - x/\sqrt{2\pi}$, the claim follows with
$c_0 = 1/\sqrt{2\pi(1-\alpha)\beta}$.
\end{proof}

\begin{corollary}[expected rank inside the band]\label{cor:band}
Fix the $k$-th true neighbor $o^*_k$ and let the \emph{band} be
\[
\mathcal{B} \;=\; \Bigl\{\, o \in D :
\omega(q,o) \ge \omega(q,o^*_k) - \tfrac{c}{\sqrt{\kb}} \,\Bigr\},
\]
with Assumption~\ref{as:generic} holding for each pair $(o^*_k, o)$,
$o \in \mathcal{B}$. Then by Lemma~\ref{lem:limit} and linearity of
expectation, the expected vote rank of $o^*_k$ is at least
$(\tfrac12 - c_0 c - c_1/\sqrt{\kb})\,|\mathcal{B}|$: on average a
constant fraction of the band outranks the $k$-th true neighbor, and a
verification budget that covers it must scale with $|\mathcal{B}|$.
Verification pays for the band.
\end{corollary}

\subsection{What this does and does not establish}\label{sec:evidence-scope}

Theorem~\ref{thm:vote} and Lemma~\ref{lem:limit} bracket the \emph{raw
vote}: it resolves overlap gaps above $\Theta(1/\sqrt{\kb})$ and not below.
They do not assert that no other processing of the stored evidence ---
ranks, idf weights, quantized posting distances --- can do better; that is
an empirical question, and our experiments so far bear only on the
families we tried (\S\ref{sec:negative}): statistical re-weightings,
navigation over the implicit graph, and triangle-inequality bounds, each
of which failed to move the recall/cost frontier on SIFT, GloVe and Deep.
Two designed families remain untested: a learned re-ranker over the stored
features, and distance \emph{estimation} (rather than bounding) from
personal pivots. We therefore state the empirical finding with its scope
--- \emph{on these datasets, no query-time refinement we tried beats
vote-then-verify} --- and reserve the word ``barrier'' for
Lemma~\ref{lem:limit}, which is what is actually proved.

\begin{remark}[no bound without a stored scale]\label{rem:noscale}
One refinement can be ruled out outright: a distance \emph{bound} on
$d(q,o)$ computed from the stored evidence plus the query's known
distances to its terms. The index stores identities and ranks --- data
that is ordinal in $o$ --- and ordinal data carries no scale. Concretely,
for any candidate $o \notin S$ replace it by $o'$ with
$d(o',x) = D + \tfrac12\,d(o,x)$ for large $D$: this is again a metric
($x \mapsto d(o,x)/2$ is $\tfrac12$-Lipschitz, so the triangle
inequality survives), it preserves the order of $o$'s distances --- hence
$o'$ has the same signature, the same build ranks, the same shared terms,
and every other object's signature is untouched --- while
$d(q,o') = D + d(q,o)/2$ is arbitrary. No function of what the index
stores and the query knows can bound $d(q,o)$.

The impossibility also names its own minimal repair: any per-object scale
anchor restores bounds, and the cheapest is the \emph{personal radius}
$\rho_o = d(o, s_{\kb})$, the distance to $o$'s last signature slot ---
one quantized byte per object, available for free at build time. Since
$d(o,s) \le \rho_o$ for every shared term and $d(o,s) \ge \rho_o$ for
every sample point absent from $o$'s signature,
\[
d(q,o) \;\le\; \min_{s\ \mathrm{shared}} d(q,s) + \rho_o,
\qquad
d(q,o) \;\ge\; \max\Bigl(\;\max_{s\ \mathrm{shared}} d(q,s) - \rho_o,\;\;
\rho_o - \min_{s\ \mathrm{absent}} d(q,s)\Bigr),
\]
where ``absent'' ranges over the query's terms not shared with $o$. The
second lower bound uses \emph{absence} as evidence: a high-vote candidate
that misses the query's nearest terms is certifiably far.
(Quantization rounds $\rho_o$ down for the lower bounds and up for the
upper.) Figure~\ref{fig:bounds} illustrates both halves. Whether these
$\rho$-scale bounds prune anything on high-LID data is an empirical
question, answered in \S\ref{sec:experiments}.
\end{remark}

\begin{figure}[t]\centering
\includegraphics[width=\linewidth]{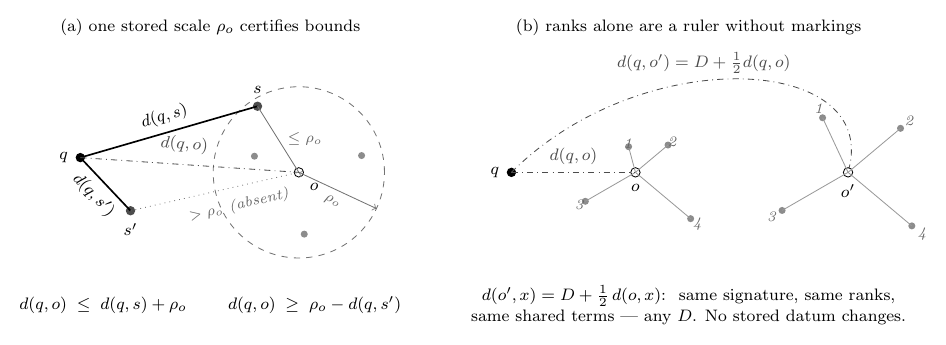}
\caption{The two halves of Remark~\ref{rem:noscale}. \emph{(a)} With one
stored scale --- the personal radius $\rho_o$, the distance from $o$ to
its last signature slot --- every shared term $s$ certifies
$d(q,o) \le d(q,s) + \rho_o$, and every term of $q$ \emph{absent} from
$o$'s signature lies outside the ball, certifying
$d(q,o) \ge \rho_o - d(q,s')$: a high-vote candidate that misses the
query's nearest terms is provably far. \emph{(b)} Without the scale no
bound exists: replacing $o$ by $o'$ with $d(o',x) = D + \tfrac12\,d(o,x)$
preserves the order of $o$'s distances --- the same signature, the same
build ranks, the same shared terms, with every other signature untouched
--- while $d(q,o')$ grows with $D$ without limit. Ranks are a ruler
without markings; the sketch is schematic since the construction is
metric, not Euclidean.}
\label{fig:bounds}
\end{figure}

\begin{remark}[truncated signatures]\label{rem:truncation}
The implementation stores the $\kb$ nearest sample points of $x$, i.e.\
$S \cap N_{R(x)}(x)$ where $R(x)$ is the neighbor rank of the $\kb$-th
nearest sample point. $R(x)$ is a negative-binomial stopping rank with
$\Pr[|R(x)/\rho - 1| > \varepsilon] \le 2e^{-\kb\varepsilon^2/8}$ for
$\varepsilon \le 1/2$ by multiplicative Chernoff bounds, so every statement above transfers to the truncated
signature with $\omega$ perturbed by $O(\varepsilon)$ and an additive
$e^{-\Omega(\kb\varepsilon^2)}$ in the failure probability.
\end{remark}

\begin{remark}[weighted votes]\label{rem:idf}
For a weighted vote $X_w(q,o) = \sum_u w_u B_u\,\mathbf{1}[u \in
N_\rho(q) \cap N_\rho(o)]$ with weights $w_u \in [w_{\min}, w_{\max}]$
fixed by the index (idf), Lemma~\ref{lem:sep} holds with $\kb\Delta$
replaced by the weighted mean gap and the same Bernstein argument,
degrading by the ratio $w_{\max}/w_{\min}$. Whether idf \emph{improves}
the effective gap $\Delta$ is a property of the data (it does on hub-heavy
data, \S\ref{sec:experiments}); the guarantee form is unchanged.
\end{remark}

\section{What query-time processing does not add}\label{sec:negative}

Five families of query-time refinements were implemented and measured.
All five failed to improve the recall/cost frontier of the static
$(\ks, C)$ configuration. Scope: these are measurements on SIFT1M,
GloVe-200 and Deep prefixes --- dense vectors of moderate-to-high
intrinsic dimension --- with the stated configurations; none is a
theorem, and two designed families (a learned re-ranker over stored
features; distance estimation from personal pivots) remain untested.

\begin{observation}[statistical]\label{obs:stat}
Down-weighting votes by build/query rank agreement loses recall in
proportion to the decay's aggressiveness and converges to idf from below
as it softens; normalizing votes by the hubness of a candidate's
signature also loses. Restricting the \emph{scan} by rank
(\texttt{max\_rank}) trades recall for merge cost smoothly and is kept
as a cost knob.
\end{observation}

\begin{observation}[navigational]\label{obs:nav}
The forward and inverted files define an implicit graph, materializable
on demand. Best-first expansion over it, interleaved with the merge,
contributed 5.1\% of final top-10 entries on GloVe while 1 candidate in
102 was unreachable through the merge ranking --- a real but marginal
contribution that did not move the frontier at equal cost.
\end{observation}

\begin{observation}[geometric]\label{obs:geom}
Storing quantized distances $d(o,s)$ (1 byte per posting) yields
per-candidate triangle bounds. The bounds are safe and nearly inert:
0--3\% of verifications pruned, and upper-bound-ordered verification is
strictly worse than vote-ordered (SIFT: 0.47 versus 0.98 recall at equal
budget). The one-byte-per-object personal-radius variant of
Remark~\ref{rem:noscale} was also measured: it prunes 0.8\% on SIFT and
0.06\% on GloVe at $C$ in the thousands, with mean relative slack
0.80--0.90, and both bound-orderings lose to vote order at every budget
tried. The impossibility half of the remark stands; the repair is
cheap and, on this data, not useful.
\end{observation}

\paragraph{Adaptive stopping.} Processing lists nearest-first and
stopping when the ranking stabilizes was simulated at three signal
strengths (vote-set stability, tolerant stability, verified-top-$k$
patience with cached distances). Every variant lands on or below the
static $(\ks, C)$ frontier; an oracle stop shows $1.8$--$2.2\times$
headroom that no cheap signal reached. A calibrated sequential test on
the vote gap --- the one variant the model does not forbid --- is
untested and listed as future work.

\paragraph{Evidence placement (posting subsampling).} Reading a random
$\beta$-fraction of each list is not equivalent to the native dial
$\kb' = \beta\kb$: thinning keeps sparse evidence at the full radius
$\rho$, the dial keeps dense evidence at radius $\beta\rho$. Measured at
equal postings read (Table~\ref{tab:placement}), the native truncation
dominates everywhere, and thinned-32 loses to native-16 while reading
twice the postings. Under a posting budget, the nearest slots are the
right ones; approximate merging by subsampling is dominated by moving
the dials.

\begin{table}[t]\centering
\caption{Evidence placement at equal postings read (SIFT1M, $\kb{=}64$
build, identical query terms, seed~1; seed~2 replicates within
$\pm0.002$). ``native-$k$'': the $k$ nearest signature slots;
``thinned-$k$'': a random $k$ of the 64.}
\label{tab:placement}
\begin{tabular}{lrrr}
\toprule
arm & postings/query & recall@$C{=}1000$ & recall@$C{=}5000$ \\
\midrule
full-64 & 351,668 & 0.9521 & 0.9941 \\
native-32 & 175,736 & 0.9418 & 0.9924 \\
thinned-32 & 175,820 & 0.8922 & 0.9824 \\
native-16 & 87,059 & 0.9174 & 0.9875 \\
thinned-16 & 87,883 & 0.7695 & 0.9392 \\
\bottomrule
\end{tabular}

\end{table}

\section{Vocabulary selection}\label{sec:coverage}

Sample placement cannot change the mean list length $\rho$; it changes
the distribution, and on hub-heavy data that distribution is wild. The
re-sampling heuristic evaluated here retires the most-loaded samples and
promotes members of the most idle lists ($\delta = 0.15$, two
iterations, rebuild each time; direction matters --- splitting hubs
loses).

Table~\ref{tab:coverage} reports the effect over five sample seeds. On
SIFT the gain is large at small budgets ($+\num{sift/cov-gain@C500}$
recall at $C{=}500$, far beyond seed noise of $\pm 0.002$) and decays to
nothing at $C{=}20{,}000$. On GloVe the gain at small $C$ is under two
points and the effect \emph{reverses} at large $C$ (about $-0.004$ at
$C{=}50{,}000$). The summary supported by the seeded data: coverage
re-sampling helps exactly where the verification budget is the binding
constraint, is neutral-to-harmful where the budget is generous, and
costs two extra builds. Single-run measurements overstate it; the seeds
are what reveal the sign reversal.

\begin{table}[t]\centering
\caption{Coverage re-sampling with seeds: recall@10, mean $\pm$ std over
5 sample seeds, $\ks{=}64$, idf scoring.}
\label{tab:coverage}
\begin{tabular}{llrr}
\toprule
dataset & $C$ & uniform & 2 spread iterations \\
\midrule
sift & 500 & 0.9052$\pm$0.0017 & 0.9425$\pm$0.0007 \\
sift & 2,000 & 0.9790$\pm$0.0006 & 0.9910$\pm$0.0004 \\
sift & 20,000 & 0.9991$\pm$0.0000 & 0.9993$\pm$0.0000 \\
\midrule
glove & 1,000 & 0.6956$\pm$0.0015 & 0.7145$\pm$0.0012 \\
glove & 5,000 & 0.8357$\pm$0.0020 & 0.8444$\pm$0.0010 \\
glove & 50,000 & 0.9507$\pm$0.0008 & 0.9469$\pm$0.0010 \\
\bottomrule
\end{tabular}

\end{table}

\section{The merge}\label{sec:merge}

The merge scores candidates from $\ks$ posting lists --- the
$t$-occurrence problem \cite{li2008divideskip}. The production
implementation is a dense per-thread scoreboard with touched-entry
reset; a radix-partitioned variant with an $O(\ks\rho)$ transient
working set serves the disk-resident and large-$n$ regimes with
bit-identical output. Five implementations (scoreboard, radix,
stratified branch-and-bound, quotiented, generate-and-score) read the
same posting volume within constant factors; their differences are
locality and representation. Two remarks bound what is possible:

\emph{Reading.} Within-list weights are uniform (count and idf attach to
the list, not the posting), so any \emph{exact} top-$C$ merge on this
layout must read every posting of the selected lists: for a skipped
posting, one consistent completion places its object one vote below the
cutoff and another does not. This is scoped to rank-major,
id-unordered lists and uniform weights: on id-sorted layouts the
skipping algorithms of \cite{li2008divideskip} apply in principle, and
their measured skip rates in our regime (many lists, low thresholds) are
work in progress; non-uniform within-list weights would re-enable
WAND-style pruning at a recall price Observation~\ref{obs:stat}
measures.

\emph{Accounting.} Total postings are $\kb n$ over $\alpha n$ lists, so
a query streams $W = h\,\ks\rho$ postings, with $h \ge 1$ a hubness
factor ($h = 1.72$--$1.75$ measured on SIFT across
$\alpha = 1$--$4\%$), and $W$ is independent of $n$ at fixed dials. Per
query, the identity $W \times |S| = h\,\ks\kb n$ says that halving
residency doubles the stream. Batching refines what is conserved: for a
batch of $B$ queries the posting matrix can be streamed once and
accumulated as a tiled sparse product, so posting-\emph{read} traffic
per query falls as $\kb n / B$ beyond $B^\ast = m/\ks$ queries, while
the vote \emph{additions} $B\ks\rho$ are irreducible --- they are the
votes. The conserved quantity is operations, not bytes; an
implementation and measurement of the tiled batch merge is future work.

Query anatomy at the measured configurations: selecting lists is
$\le 4\%$ of query time, verification is bounded by $C$, and 75--93\% is
the accumulate loop, compute-bound through $10^7$ and cache-bound at
$10^8$ on a 4-socket machine. The merge is where the design pays for its
properties, and \S\ref{sec:position} discusses what does and does not
help.

\section{Experiments}\label{sec:experiments}

\paragraph{Protocol.} Machines: MacBook Air M5 (10 cores, 32\,GB,
fanless; QPS on this machine varies up to 25\% across identical runs
from thermal state, so recall differences are the reliable signal at
$10^6$ and throughput is reported as the median of three interleaved
repetitions), and a 4-socket Intel Xeon E7-4809\,v3 (64 threads,
1.5\,TB, NVMe; a 2015 NUMA machine --- its per-thread merge cost at
$10^8$ includes a documented remote-access penalty). Datasets: SIFT1M
(128d, $\ell_2$), GloVe-200 (angular, unit-normalized, $n =
1{,}183{,}514$), Deep 1M/10M/100M (96d, $\ell_2$, nested prefixes of
Deep1B with exact ground truth per prefix). Queries: the standard
10{,}000-query sets ($2{,}000$ for the cgroup runs). Recall\;=\;
recall@10 against exact ground truth. \misi{} runs use
$\alpha{=}2\%$, $\kb{=}64$, idf scoring, GRAFT as the inner index, and
5 sample seeds at $n \le 10^7$ unless stated; baselines are swept to a
saturation rule ($\Delta$recall $< 5{\cdot}10^{-4}$ per
$1.5\times$ step, or a latency floor). Every number in this paper is
generated from a JSON run manifest checked into the repository; tables
are emitted by script from those manifests.

\subsection{One million objects, including the ancestor}

Table~\ref{tab:e1} and Figure~\ref{fig:pareto} compare \misi{}, NAPP
(NMSLIB implementation, pivots $\in \{2, 8, 32\}{\cdot}10^3$, index/search
thresholds swept) and \texttt{hnswlib} (M $\in \{16, 32\}$, ef swept to
saturation) under one protocol.

\begin{table}[t]\centering
\caption{Best configuration per system at fixed recall bands, $10^6$
scale, M5 laptop, batch queries, all threads. Full grids in the
repository manifests.}
\label{tab:e1}
\begin{tabular}{llrrl}
\toprule
dataset & system & recall@10 & QPS & configuration \\
\midrule
sift & misi ($\ge$0.95) & 0.9504 & 4,894 & $k_s{=}16$, $C{=}2,000$ \\
sift & misi ($\ge$0.99) & 0.9901 & 2,637 & $k_s{=}16$, $C{=}5,000$, cov \\
sift & misi (max) & 0.9993 & 776 & $k_s{=}64$, $C{=}20,000$, cov \\
sift & napp ($\ge$0.95) & 0.9645 & 2,371 & $t{=}4$ \\
sift & napp ($\ge$0.99) & 0.9939 & 1,207 & $t{=}8$ \\
sift & napp (max) & 0.9994 & 373 & $t{=}2$ \\
sift & hnswlib ($\ge$0.95) & 0.9571 & 28,735 & ef${=}58$ \\
sift & hnswlib ($\ge$0.99) & 0.9954 & 15,751 & ef${=}133$ \\
sift & hnswlib (max) & 0.9992 & 3,365 & ef${=}679$ \\
\midrule
glove & misi ($\ge$0.9) & 0.9143 & 494 & $k_s{=}64$, $C{=}20,000$, cov \\
glove & misi ($\ge$0.95) & 0.9690 & 257 & $k_s{=}128$, $C{=}50,000$ \\
glove & misi (max) & 0.9803 & 149 & $k_s{=}256$, $C{=}50,000$ \\
glove & napp ($\ge$0.9) & 0.9146 & 175 & $t{=}3$ \\
glove & napp ($\ge$0.95) & 0.9614 & 136 & $t{=}4$ \\
glove & napp (max) & 1.0000 & 34 & $t{=}2$ \\
glove & hnswlib ($\ge$0.9) & 0.9175 & 1,487 & ef${=}452$ \\
glove & hnswlib ($\ge$0.95) & 0.9532 & 717 & ef${=}1019$ \\
glove & hnswlib (max) & 0.9865 & 239 & ef${=}3442$ \\
\bottomrule
\end{tabular}

\end{table}

\begin{figure}[t]\centering
\includegraphics[width=\linewidth]{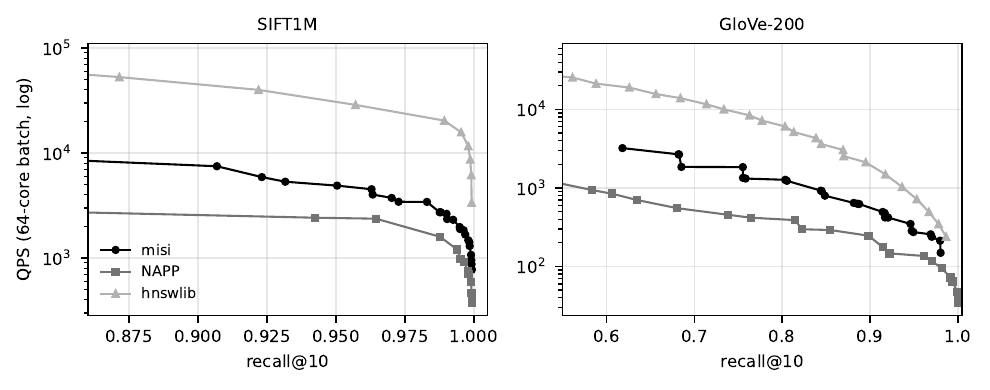}
\caption{Pareto frontiers at $10^6$ (pooled over configurations and
seeds; log QPS). \texttt{hnswlib} dominates both vote-and-verify
structures on SIFT; on GloVe the three systems interleave, and NAPP's
threshold dial reaches exact recall at 34\,QPS.}
\label{fig:pareto}
\end{figure}

Three facts. First, \texttt{hnswlib} dominates on SIFT: at $0.99$ recall
it is \num{sift/hnsw-over-misi@0.99}$\times$ faster than \misi{}
(\num{sift/hnswlib/qps@0.99} versus \num{sift/misi/qps@0.99} QPS), and
its saturated ceiling matches \misi's. On easy in-RAM data at this
scale there is no throughput or recall argument for either
vote-and-verify structure; the argument is construction time and the
properties of \S\ref{sec:position}. Second, \misi{} improves on its
ancestor by about \num{sift/misi-over-napp@0.99}$\times$ at matched
recall on SIFT --- a real but modest factor at this scale, consistent
with the analysis: the linear vocabulary's advantage is asymptotic, not
constant-factor. Third, on hub-heavy GloVe the ordering inverts in
places: NAPP's threshold degrades gracefully toward exhaustive search
and reaches $0.9917$ recall (and eventually $1.0$) at low QPS, above
both \misi's grid and \texttt{hnswlib}'s saturated $0.9865$; \misi{}
reaches $0.9803$ only at $\ks{=}256$, $C{=}50{,}000$ --- verifying
$4.2\%$ of the database. On hard data, every structure here buys its
last points of recall by approaching a scan, and the honest comparison
is who does so most gracefully.

\subsection{Scaling}

\begin{table}[t]\centering
\caption{The verification budget for $0.99$ recall over nested Deep
prefixes at fixed dials ($\alpha{=}2\%$, $\kb{=}64$).}
\label{tab:scale}
\begin{tabular}{lrrl}
\toprule
$n$ & $C_{0.99}$ & recall at $C_{0.99}$ & basis \\
\midrule
$10^6$ & 5,000 & 0.9917$\pm$0.0002 & 5 seeds, $k_s{=}64$ \\
$10^7$ & 10,000 & 0.9921$\pm$0.0003 & 5 seeds, $k_s{=}64$ \\
$10^8$ & 20,000 & 0.9936 & 1 seed, $k_s{=}128$ (uniform artifact) \\
\bottomrule
\end{tabular}

\end{table}

\begin{figure}[t]\centering
\includegraphics[width=\linewidth]{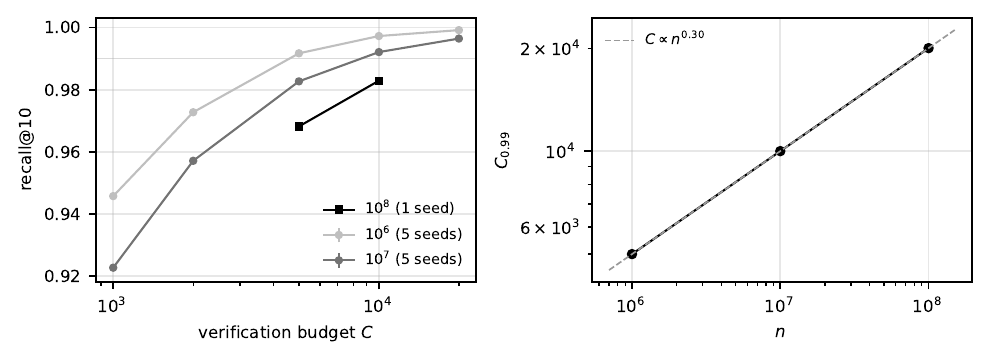}
\caption{Left: recall versus verification budget at three scales
(error bars: std over 5 seeds where present). Right: the fitted budget
law $C_{0.99} \propto n^{\num{scaling/gamma}}$.}
\label{fig:scale}
\end{figure}

Measured on nested prefixes of one distribution with seeds
(Table~\ref{tab:scale}), the budget doubles per decade:
$C_{0.99} \propto n^{\num{scaling/gamma}}$ over
$10^6$--$10^8$. The fraction of the database verified still falls by
$\sim\!5\times$ per decade, but the absolute budget grows, and the model
section's account is that the confusable count $|K(q,\omega^-)|$ grows
with $n$ at fixed resolution. The dial tests (Table~\ref{tab:dials})
move as the model predicts: recall at fixed $C$ rises with $\alpha$
(more evidence at finer resolution) and with $\kb$ (more evidence at
the same resolution), and both purchases are paid at build time.

\begin{table}[t]\centering
\caption{Dial tests at $10^7$ (seed 1, $\ks{=}64$).}
\label{tab:dials}
\begin{tabular}{llrrr}
\toprule
dial & value & recall@$C{=}5000$ & recall@$C{=}20000$ & build (s) \\
\midrule
-- & $\alpha{=}2\%$, $k_b{=}64$ & 0.9829 & 0.9963 & 312 \\
$\alpha$ & 1\% & 0.9776 & 0.9953 & 242 \\
$\alpha$ & 4\% & 0.9864 & 0.9969 & 399 \\
$\alpha$ & 8\% & 0.9876 & 0.9965 & 515 \\
$k_b$ & 16 & 0.9635 & 0.9878 & 109 \\
$k_b$ & 32 & 0.9765 & 0.9938 & 211 \\
$k_b$ & 128 & 0.9851 & 0.9969 & 600 \\
\bottomrule
\end{tabular}

\end{table}

\begin{figure}[t]\centering
\includegraphics[width=0.55\linewidth]{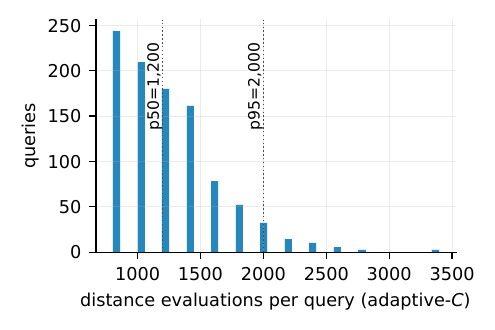}
\caption{Per-query verification budgets under patience-based adaptive
stopping (Deep-10M, single seed): the spread between easy (p10 $=$ 800)
and hard (max $=$ 3{,}400) queries. In the model's terms this is the
per-query confusable count $|K(q,\omega^-)|$ made visible; the adaptive
rule prices it without a tuned $C$.}
\label{fig:hardness}
\end{figure}

\subsection{Saturated ceilings at $10^8$}

\begin{table}[t]\centering
\caption{Maximum recall reached by each system at $10^8$ when its search
dial is swept to the saturation rule. \misi{} rows are single-seed.
DiskANN's sweep was limited by the largest $L$ run, not by saturation.}
\label{tab:sat}
\begin{tabular}{lrrll}
\toprule
system & max recall@10 & QPS there & config & saturated? \\
\midrule
GRAFT (quality, in RAM) & 0.9992 & 506 & ef${=}4000$ & yes \\
HNSW (M32, in RAM) & 0.9989 & 540 & ef${=}4000$ & borderline \\
DiskANN (SSD, PQ 9.6\,GB) & 0.9984 & 568 & $L{=}800$ & no (sweep limit) \\
misi (in RAM, coverage) & 0.9967 & 175 & $k_s{=}128$, $C{=}20,000$ & no (dial limit) \\
\bottomrule
\end{tabular}

\end{table}

Recall ceilings are a recurring claim in this literature, and sweeping
every dial to an explicit saturation rule is the only way to evaluate
one. Under such sweeps (Table~\ref{tab:sat}), GRAFT's quality
configuration reaches \num{sat/graft-max}, HNSW \num{sat/hnsw-max}, and
DiskANN \num{sat/diskann-max}, all above \misi's \num{sat/misi-max}, at
throughputs between $506$ and $568$ QPS --- within a factor of $3.2$ of
\misi's $175$ at its own maximum. No system here holds a ceiling
advantage: at $10^8$ all four converge within half a point of exact at
low throughput, and the differences that matter at this scale are
construction cost, residency, and updates, not the last digit of
recall.

\subsection{Memory-limited serving}

\begin{table}[t]\centering
\caption{Serving Deep-100M under enforced cgroup memory budgets
(\texttt{MemoryMax}, swap off, inputs evicted, postings and vectors on
NVMe; $2{,}000$ queries; \misi{} uniform artifact, disk residency,
radix merge). DiskANN: prebuilt R64/L100, PQ codes $9.6$\,GB, beam
$W{=}4$, no node cache.}
\label{tab:memory}
\begin{tabular}{llrrrr}
\toprule
budget & system & recall@10 & QPS & p50 (ms) & p99 (ms) \\
\midrule
2--4\,GB & misi & \multicolumn{4}{l}{OOM at load (inner-index build)} \\
2--8\,GB & DiskANN & \multicolumn{4}{l}{OOM (PQ table alone is 9.6\,GB)} \\
\midrule
8\,GB & misi ($k_s{=}64$, $C{=}5,000$) & 0.9682 & 30 & 669 & 728 \\
8\,GB & misi ($k_s{=}64$, $C{=}10,000$) & 0.9829 & 16 & 1226 & 1302 \\
8\,GB & misi ($k_s{=}128$, $C{=}20,000$) & 0.9936 & 8 & 2513 & 2723 \\
\midrule
16\,GB & misi ($k_s{=}64$, $C{=}10,000$) & 0.9829 & 20 & 1040 & 1144 \\
16\,GB & misi ($k_s{=}128$, $C{=}20,000$) & 0.9936 & 10 & 2121 & 2337 \\
12\,GB & DiskANN ($L{=}50$) & 0.9315 & 6179 & 10\,(mean) & -- \\
12\,GB & DiskANN ($L{=}200$) & 0.9874 & 1983 & 32\,(mean) & -- \\
12\,GB & DiskANN ($L{=}800$) & 0.9984 & 568 & 111\,(mean) & -- \\
\bottomrule
\end{tabular}

\end{table}

A memory claim is only as good as the enforcement behind it, so
Table~\ref{tab:memory} is measured under cgroup limits with swap off
and inputs evicted, and it corrects the naive arithmetic in both
directions. \misi's measured serving floor is between $4$ and
$8$\,GB --- not the $0.78$\,GB of sample vectors --- because the inner
index is rebuilt over the
$2{\cdot}10^6$-point sample at load time, and that build, not the
serving state, sets the peak. Within an $8$\,GB budget \misi{} serves
$10^8$ vectors at \num{e5/misi-8g-qps-low}--\num{e5/misi-8g-qps-high}
QPS with second-scale latencies; DiskANN does not run at $8$\,GB (its
PQ table alone is $9.6$\,GB) and at its native $12$\,GB serves the same
recall band at $\sim\!100\times$ the throughput
(\num{e5/diskann-12g-qps@0.987} QPS at $0.9874$), because it reads
$\sim\!10^2$ disk locations per query where \misi{} reads
$C \approx 10^4$. The narrow conclusion supported by measurement: there
is a real 4-GB-wide window in which \misi{} answers and the SSD-graph
baseline cannot, at single-digit-to-tens QPS; below that window neither
runs; above it DiskANN dominates throughput. The composition that would
change this --- PQ codes resident as a verification screen, cutting the
random reads to $\sim\!10^2$ --- is designed but not implemented, and
this paper claims nothing for it.

\subsection{Construction, portability, streaming}

\begin{table}[t]\centering
\caption{Construction wall-clock, all threads, same machine per row.
Coverage re-sampling costs two additional builds where used.}
\label{tab:build}
\begin{tabular}{lrl}
\toprule
dataset (machine) & misi build & baselines, same machine \\
\midrule
SIFT1M (laptop) & 12--24\,s & GRAFT 111\,s; hnswlib M32 152\,s; NAPP 24\,s ($m{=}$2k) to 1048\,s ($m{=}$32k) \\
GloVe-200 (laptop) & 137--194\,s & GRAFT 580\,s; hnswlib 462\,s; NAPP 96--2224\,s \\
Deep-10M (Xeon) & 302\,s & GRAFT 635\,s (tuned) / 1,140\,s (quality); DiskANN 23\,min \\
Deep-100M (Xeon) & 5,250\,s (87.5\,min) & GRAFT 14,765\,s (default) / 19,675\,s (quality); HNSW M32 $\sim$2.4\,h \cite{graft} \\
\bottomrule
\end{tabular}

\end{table}

Construction (Table~\ref{tab:build}) is where the design's structure
shows plainly: $n$ independent searches parallelize without
coordination, and \misi{} builds $10^8$ vectors in
\num{build/deep100m-misi-s}\,s against $19{,}675$\,s for the graph
configuration that matches its recall --- with the caveats that
coverage re-sampling, where used, spends most of that lead, and that
GRAFT's faster default configuration ($14{,}765$\,s) trails by less.

The artifact is a directory of flat arrays. As a portability check, the
Deep-10M index built on the Xeon was copied to the laptop, loaded in
$15.7$\,s, and served at $862$\,QPS at $0.9911$ recall --- matching the
4-socket server at equal recall, per-core locality beating $64$
contended cores on this workload. Construction also runs fully
streamed: signatures are computed in chunks against the frozen inner
index and the inverted file is produced by an external sort, with
output bit-identical to the in-memory build; under an enforced
$3$\,GiB cgroup cap (inputs evicted, peak read from the cgroup) the
Deep-10M factory build completes in $462$\,s, an $8\%$ penalty over
uncapped. The $10^8$ factory build wrote its artifact in $6{,}726$\,s.

\subsection{Beyond $\ell_2$, and a GPU port}

With a brute-force inner index the pipeline consumes the dissimilarity
only through induced rankings plus final evaluations. Dot-product
similarity on Deep-5M indexes directly ($0.9930$ recall@10). Under
adversarial MIPS (lognormal norms) global hubs make signatures alike
(posting skew $535$) and recall collapses to $0.28$; coverage
re-sampling does not move it, because the skew is similarity-induced
rather than sample-induced; the Neyshabur--Srebro reduction restores
$0.9907$ (skew $37$). The design composes with reductions as with inner
indexes; no non-vector metric dataset is evaluated here, and the title's
scope should be read accordingly: metric by construction, evaluated on
vectors.

A three-kernel CUDA port (sketch accumulation, survivor collection,
exact scoring) serves the Deep-5M index at $1{,}139$\,QPS at $0.9926$
recall on a 65\,W laptop GPU (RTX 3060, $\sim\!4$\,GB VRAM),
$4.1\times$ the same machine's 16-thread CPU. No GPU ANN baseline was
run --- CAGRA and FAISS-GPU on the same card are required before any
comparative claim, and the measurement is reported only as evidence
that the counting kernel ports.

\section{Position}\label{sec:position}

Where the measurements place this design, stated without advocacy:

\emph{Against its ancestor}, \misi{} is NAPP with the two scale walls
removed. At $10^6$ that is worth about
\num{sift/misi-over-napp@0.99}$\times$; the structural difference is
that NAPP at $10^8$ with $m = 32{,}000$ pivots would need
$3.2{\cdot}10^{12}$ brute-force pivot distances at build and mean lists
of $10^5$ postings, while \misi's lists stay at $\rho = 3{,}200$ and
its pivot search is delegated. The lineage claim is architectural, and
the constant factors at small scale are honest about it.

\emph{Against in-RAM graphs}, \misi{} loses throughput at matched
recall by $6$--$16\times$ at $10^6$ and, against saturated
configurations, holds no recall-ceiling advantage at any scale
measured. It wins construction time
($\num{build/ratio-graft-quality}\times$ at $10^8$ against the
matched-quality build), determinism, and artifact portability.

\emph{Against the SSD graph}, \misi{} serves in a measured $4$-GB
window where DiskANN cannot, and loses throughput by two orders of
magnitude at DiskANN's native budget. The verification I/O gap
($10^4$ versus $10^2$ reads per query) is the whole story, and a
resident PQ screen is the obvious, unimplemented reply.

\emph{The merge} is the cost center (75--93\% of query time), and its
volume is fixed by the reading argument of \S\ref{sec:merge} for exact
scoring on this layout. What remains open, in order of theoretical
leverage: skipping on id-sorted layouts (measured skip rates pending),
learned within-list impacts (would re-enable WAND at a recall price not
yet mapped), calibrated sequential stopping, tiled batch execution
(bytes amortize; operations do not), and bandwidth-rich substrates ---
the GPU measurement above is a first point on that curve, and bank-level
processing-in-memory DRAM is a speculative later one: under the same
id-block layout the accumulate is a per-bank local histogram and a
verification is an in-bank dot product, though no such device has run
this workload.

\section{Limitations and future work}

The update path is a design, not an implementation. Non-vector metrics
are unevaluated. The SPANN comparison is absent. All $10^8$ \misi{}
rows are single-seed. The 4--8\,GB serving floor is an artifact of
rebuilding the inner index at load; a streamed or serialized inner
index would lower it, and its size is the binding residency at every
scale --- at $10^{12}$, the sample itself outgrows RAM, and recursive
construction (this index over its own sample) is the design's own
answer, gated on inner-index recall at $m \ge 10^8$, which is
unmeasured. A resident PQ verification screen is the change most likely
to alter the memory-limited serving table. The coverage heuristic needs
a selection principle rather than a direction that happens to win; the
per-query confusable count that Theorem~\ref{thm:vote} makes measurable
is the natural target for calibrated stopping. These are stated as open
because they are.

\paragraph{Reproducibility.} Code, the \texttt{misifu} CLI, every run
manifest, and the scripts that generate every table and figure in this
paper from those manifests:
\url{https://github.com/zevahcle/MISIFU}.

\bibliographystyle{abbrv}

\end{document}